\documentclass[]{article}

\usepackage{amsthm,amssymb,amsmath}
\usepackage[margin=1in]{geometry}

\newtheorem{theorem}{Theorem}
\newtheorem{lemma}[theorem]{Lemma}

\newtheorem{corollary}[theorem]{Corollary}

\newcommand{\NP}{\normalfont{NP}}

\newcommand{\FPT}{\normalfont{FPT}}

\newcommand{\Card}[1]{|#1|}
\newcommand{\Yes}{\textsc{Yes}}

\newcommand{\SB}{\{\,}
\newcommand{\SM}{\;{|}\;}
\newcommand{\SE}{\,\}}

\newcommand{\pcproblem}[4]{\begin{samepage}\begin{quote} \textsc{#1}\\ 
\textit{Instance:} #2\\ \textit{Parameter:} #3\\ \textit{Question:} #4 \end{quote}\end{samepage}}
\newcommand{\classicalproblem}[3]{\begin{samepage}\begin{quote} \textsc{#1}\\ 
\textit{Instance:} #2\\ \textit{Question:} #3 \end{quote}\end{samepage}}

\DeclareMathOperator{\tw}{tw}
\DeclareMathOperator{\ltw}{ltw}

\begin{document}

\title{On the Treewidth of Dynamic Graphs}
\author{Bernard Mans \and Luke Mathieson}
\date{}

\maketitle

\begin{abstract}
\emph{Dynamic} graph theory is a novel, growing area that deals with graphs that change over time and is of great utility in modelling modern wireless, mobile and dynamic environments. As a graph evolves, possibly arbitrarily, it is challenging to identify the graph properties that can be preserved over time and understand their respective computability. 

In this paper we are concerned with the {\em treewidth} of dynamic graphs. We focus on 
\emph{metatheorems}, which allow the generation of a series of results based on general properties of classes of structures. In graph theory two major metatheorems on treewidth provide complexity classifications by employing structural graph measures and finite model theory. Courcelle's Theorem gives a general tractability result for problems expressible in monadic second order logic on graphs of bounded treewidth, and Frick \& Grohe demonstrate a similar result for first order logic and graphs of bounded local treewidth.

We extend these theorems by showing that dynamic graphs of bounded (local) treewidth where the length of time over which the graph evolves and is observed is finite and bounded can be modelled in such a way that the (local) treewidth of the underlying graph is maintained. We show the application of these results to problems in dynamic graph theory and dynamic extensions to static problems. In addition we demonstrate that certain widely used dynamic graph classes naturally have bounded local treewidth.
\end{abstract}

\section{Introduction}

Graph theory has proven to be an extremely useful tool for modelling computational systems and with the advent and growing preponderance of mobile devices and dynamic systems  it is natural that graph theory is extended and adapted to capture the evolving aspects of these environments. Dynamic graphs have been formalised in a number of ways: e.g., time-varying graphs~\cite{CasteigtsFlocchiniMansSantoro10,CasteigtsFQS11,FlocchiniMS09}, Carrier-based networks~\cite{BrejovaDKV11}, evolving graphs~\cite{CasteigtsCF09,Ferreira04,XuanFerreiraJarry03}, dynamic networks~\cite{KuhnO11,KuhnLO10}, scheduled networks~\cite{Berman96}, temporal networks~\cite{KKK00}, opportunistic networks~\cite{ChaintreauMMD07,JacquetMMR08}, Markovian~\cite{ClementiMS11}.
When considering the dynamic aspects of a dynamic graph, even classically simple properties such as shortest paths become more complex to compute and may even become definitionally ambiguous~\cite{XuanFerreiraJarry03}.

In this paper, we are not particularly interested in any particular dynamic model. Initially we will loosely use the term \emph{dynamic graph} and for our purpose we will define the term formally using the simplest possible definition as it is a generalized and reasonably assumption free model for a dynamic graph. For this paper, one of the key questions in moving from static graph theory to dynamic graph theory concerns the preservation of properties and their computability.

One important general approach to the complexity and computability of properties in (static) graphs is the application of \emph{metatheorems} which classify large classes of problems. Arguably the most famous of these metatheorems is \emph{Courcelle's theorem} (stated and proved over a series of articles from~\cite{Courcelle90} to~\cite{Courcelle06}) which gives a polynomial time algorithm for any monadic second-order expressible property on graphs of bounded treewidth. More precisely the model checking problem for monadic second-order logic is fixed-parameter tractable with parameter $\Card{\phi}+k$ where $\phi$ is the sentence of logic and $k$ is the treewidth of the input structure. Frick and Grohe~\cite{FrickGrohe01} give a similar metatheorem for first-order logic and structures of locally bounded treewidth, which allows a greater class of structures (all structures of bounded treewidth have locally bounded treewidth), but constrains the logical expressibility. In fact Dvo\v{r}\'{a}k \emph{et al.}~\cite{DvorakKT10} show that properties expressible in first-order logic are decidable in linear time for graphs of bounded expansion, a superclass of several classes of sparse graphs, including those with bounded local treewidth.
Stewart~\cite{Stewart08} demonstrates that Frick and Grohe's result holds if the bound on the local treewidth is dependent on the parameter of the input problem, rather than simply constant.

Such metatheorems are extremely useful classification tools, and having them available for use in the context of dynamic graphs would be highly desirable. Two questions immediately arise when considering such an extension --- are there any restrictions that have to be made to the logic and are there any further constraints on the structures (in this case the graphs)? Answering the first question is simple: no. An apparent natural match for dynamic graphs is temporal logic (a form of modal logic), however temporal logic has a simple canonical translation into first-order and monadic second-order logic. The addition of a ``time'' relation, that provides a temporal ordering is sufficient to capture the ideas of temporal logic. Thus we are not limited (more than before) with the logics that are applicable. With regards to the structure, we must first consider the setting of these metatheorems. They are both true in the context of finite model theory, emphasizing the \emph{finite}. Therefore we are immediately limited to finite temporal domains (a finite domain for the vertices and hence edges of the graph is expected of course). Beyond this if we are to capture the temporal aspect in the structure, we must do so in a manner that is both useable and respects the constraints on the structure necessary for the metatheorems --- namely bounded treewidth and local treewidth respectively.

Dynamic problems have been approached before in a number of ways. Hagerup~\cite{Hagerup00} examines the situation where the logical expression changes (\emph{i.e.} the question being asked changes), and Bodlaender~\cite{Bodlaender93}, Cohen \emph{et al.}~\cite{CohenSTV93} and Frederickson~\cite{Frederickson98} give a variety of results that deal with graphs that change, requiring changes in the tree decompositions used. In contrast this work deals with problems that include the temporal dimension as an intrinsic aspect, rather than problems that can be answered at an instant in time (and then may need to be recomputed). A simple example is a \emph{journey}, the temporal extension of a path, where the route between two vertices may not exist at any given point of time, but as the graph changes the schedule of edges may allow traversal.

The remainder of this paper deals with the translation of dynamic graphs into structures that maintain these bounds if the original graph did so at every point in time. In addition we demonstrate the utility of these metatheorems in classification by application to some open problems on dynamic graphs.

\section{Preliminaries}

\subsection{Graphs and Dynamic Graphs}

The graphs we employ are simple graphs, both directed and undirected. For the usual graph properties we use the usual notation (\emph{qq.v.} Diestel~\cite{Diestel2000}). 

We define a \emph{dynamic graph} for our purposes as a (di)graph augmented with a set $T$ of times, a function $\zeta_{e}$ which maps the edge set to a subset of $T$ and a function $\zeta_{v}$ that maps the vertex set to subsets of $T$, representing the times at which the edge or vertex exists. We will consider a discrete and finite set $T$ of times for two main reasons: (i) typically, in most applications at hand, a continuous set $T$ could be easily made discrete when considering time of changes of the graph as an event-based model; (ii) similarly, applications of interest are not focussed on the observation of the evolution of the graph over an infinite period, but rather on some sufficiently large, yet finite, period where some expected or asymptotic behaviour can be identified. 

More formally a dynamic graph $G=(V,E,T,\zeta_{v},\zeta_{e})$ consists of
\begin{itemize}
\item A set $V$ of vertices.
\item A set $E\subset V\times V$ of edges (where $vv \notin E$ for any $v \in V$).
\item A set $T$ of times.
\item A function $\zeta_{v}: V \rightarrow 2^{T}$.
\item A function $\zeta_{e}: E \rightarrow 2^{T}$.
\end{itemize}

We denote the static graph derived from a dynamic graph $G$ by taking the snapshot at time $t$ as $G_{t}$. We do not allow edges to exist when either of their vertices do not, although the case where they do could be of interest in modelling, for example, wireless networks with long transmission ranges. We also assume that there is some order $<_{T}$ defined over $T(G)$, however except where specified we do not assume that this order has any particular properties (\emph{i.e.} we assume neither totality nor transitivity).

This dynamic graph definition is independent from other referenced definitions, though relationships can be drawn. Again, here, we aim to focus on the graph theoretic aspects, rather than modelling a particular system, thus we attempt to keep the definition as simple and open as possible. We also note that a similar definition with a continuous notion of time can be converted to an \emph{event-based} model, which then fulfils the requirements of our definition.

\subsection{Treewidth and Local Treewidth}

A tree decomposition $\mathbb{T}$ of a graph $G$ is a tree $\mathcal{T}$ and mapping $\beta: V(\mathcal{T}) \rightarrow 2^{V(G)}$ with the following properties:

\begin{itemize}
\item $\bigcup_{t\in V(\mathcal{T})}\beta(t) = V(G)$.
\item For each edge $uv \in E(G)$ there is some $t\in V(\mathcal{T})$ such that $u,v \in \beta(t)$.
\item For each $v \in V(G)$, the set $\SB{} t\in V(\mathcal{T}) \SM{} v \in \beta(t) \SE{}$ of tree vertices induces a connected subgraph of $\mathcal{T}$.
\end{itemize}

As is standard, we call $\beta(t)$ for each $t$ a \emph{bag}.

The width of a decomposition is defined as $\max\SB{} \Card{\beta(t)} -1 \SM{} t \in V(\mathcal{T}) \SE{}$, then the treewidth of a graph $G$ (denoted $\tw(G)$) is the minimum width of a decomposition of $G$ over all possible decompositions. A class of graphs then has \emph{bounded} treewidth if there exists some constant $k$ (or parameter in the context of parameterized complexity) such that every graph in the class has treewidth at most $k$.

Let $N^{G}_{r}(v)$ denote the $r$-neighborhood of a vertex $v$ in the graph $G$. Given a graph $G$ and an integer $r$ the local treewidth (denoted $\ltw(G,r)$) is $\max\SB\tw([N^{G}_{r}(v)])\SM v\in V\SE$. A class of graphs has \emph{bounded local treewidth} if there exists some function $f$ such that for every $r$ every graph in the class has local treewidth at most $f(r)$. A class has \emph{effectively} bounded local treewidth if $f$ is also computable. Graph classes that have bounded local treewidth (but not necessarily bounded treewidth) include graphs of bounded degree and planar graphs.

Both treewidth and local treewidth can be transferred to general relational structures by taking the (local) treewidth of the Gaifman graph (q.v. Section~\ref{subsec:prelims_logics}).

\subsection{Parameterized Complexity}

Parameterized Complexity augments combinatorial problems using parameters as independent measures of structure in addition to the overall size of the input. An instance $(I,k)$ of a parameterized problem consists of the input $I$, corresponding to the input of a classical problem and an integer parameter $k$, a special part of the input independent from $\Card{I}$.

A problem is \emph{fixed-parameter tractable} (or in $\FPT{}$) if there is an algorithm that solves each instance $(I,k)$ of the problem in time bounded by $f(k)\cdot\Card{I}^{O(1)}$ where $f$ is a computable function dependent only on $k$.

Conversely hardness for any class in the ${\normalfont W}$-hierarchy provides evidence that a problem is not fixed-parameter tractable. Hardness for such classes is typically established by \emph{parameterized reduction}, the Parameterized Complexity reduction scheme where given an instance $(I,k)$ of problem $\Pi_{1}$ an instance $(I',k')$ of problem $\Pi_{2}$ is computed in time bounded by $f(k)\cdot\Card{I}^{O(1)}$, with $k' \leq g(k)$ for some computable function $g$ and $(I,k)$ is a \Yes{}-instance if and only if $(I',k')$ is a \Yes{}-instance.

For more detail on the theory of parameterized complexity we refer the reader to \cite{FlumGrohe06}, \cite{Niedermeier06} and \cite{DowneyFellows99}.

\subsection{Logic, Relational Structures and the Model Checking Problem}~\label{subsec:prelims_logics}

For a more detailed treatment of logic and finite model theory see the monographs of Ebbinghaus \& Flum~\cite{EbbinghausFlum95} and Ebbinghaus, Flum \& Thomas~\cite{EbbinhausFlumThomas94}.

A \emph{relational structure} $\mathcal{A}$ consists of a set $A$ of elements called the \emph{universe} and a set $\tau$ of relations of finite arity called the \emph{vocabulary}, along with an interpretation of each relation in $\tau$ over $A$. That is, for each relation $R \in \tau$ of arity $r$ there is a set of tuples from $A^{r}$ that define when $R$ is true. It is sufficient for our purposes to consider only structures where $A$ is finite and $\tau$ is finite and non-empty.

\emph{First order (FO) logic} consists of a countably infinite set of \emph{(individual) variables} and a \emph{vocabulary} $\tau$, which we will assume matches the vocabulary of the structure the logic is being applied to. Formul\ae{} of first order logic are constructed from the variables, vocabulary, the atomic relation $=$, the boolean connectives $\vee$, $\wedge$ and $\neg$ and the quantifiers $\forall$ and $\exists$. A formula is a \emph{sentence} if there are no free variables, that is, all variables are bound by a quantifier.

\emph{Second order (SO) logic} extends first order logic by including a countably infinite set of \emph{relation variables}, but construction of formul\ae{} is otherwise the same. If all relation variables are unary, then the logic is \emph{monadic (MSO)}. Note that the key difference between first and second order logic is not particularly that there are relation variables, but that second order logic allows quantification over them.

The precise semantics of applying a logic to a structure are unnecessary to discuss in detail here, suffice to say that (relation) variables of the formula are mapped to (sets of) elements of the universe and the formula is evaluated via the satisfaction of the model under the vocabulary of the structure.

Given a structure $\mathcal{A}$ and a logical formula $\phi$ the basic \textsc{Model Checking} problem is simply checking whether the formula $\phi$ holds on $\mathcal{A}$:

\classicalproblem{Model Checking (MC)}{A structure $\mathcal{A}$, and a formula $\phi$.}{$\mathcal{A}\models\phi$?}

By \textsc{MC}(C,$\Phi$) we denote the model checking problem restricted to structures in the class $C$ and formual\ae{} in the class $\Phi$.
There are two particularly interesting cases for the \textsc{Model Checking} problem, summarised by the following theorems.

\begin{theorem}[Courcelle~\cite{Courcelle06}~--~\cite{Courcelle90}]
Let $C$ be any class of structures of bounded treewidth. \textsc{MC}($C$,MSO) is fixed-parameter tractable when parameterized by the treewidth of the structure and the length of the MSO sentence.
\end{theorem}

\begin{theorem}[Frick \& Grohe~\cite{FrickGrohe01}]
Let $C$ be any class of structures of bounded local treewidth. \textsc{MC}($C$,FO) is fixed-parameter tractable when parameterized by the length of the FO sentence.
\end{theorem}

Note that in both cases the length of the formula and the treewidth may themselves be dependent on other parameters.

A key tool in actually computating the solution to the model checking problem for many classes of structures is the \emph{Gaifman graph} of the structure. Given a strucure $\mathcal{A}$, the Gaifman graph $\mathcal{G^{A}}$ is the graph obtained by taking the universe $A$ of $\mathcal{A}$ as the vertex set of $\mathcal{G^{A}}$, with an edge between two vertices if they ever appear in the interpretation of any relation in $\tau$ together. Taking a graph as a structure with binary relation $E$ (the edge relation), the Gaifman graph is precisely the graph.

\section{Dynamic Graphs as Logical Structures}

Let $G$ be a dynamic graph. We give two translations of $G$ into a logical structure, with different properties and limitations appropriate to different applications.

\subsection{Local Treewidth Preserving Structure}

Let $G=(V,E,T,\zeta_{v},\zeta_{e})$ be a dynamic graph.

Let $\mathcal{A}(G)$ be the structure obtained from $G$ with universe $A$ and vocabulary $\tau$ where $A$ consists of
\begin{itemize}
\item an element $v^{t}$ for each $v \in V(G)$ if $t \in \zeta_{v}(v)$, with a function $f_{V}$ such that $f_{V}(v^{t}) = f_{V}(u^{t'})$ if and only if $v^{t}$ and $u^{t}$ are derived from the same vertex $v \in V(G)$,
\item an element for each $t \in T(G)$,
\end{itemize}
and $\tau = \SB V^{\mathcal{A}}, E^{\mathcal{A}}, T^{\mathcal{A}}, R^{\mathcal{A}}, L_{t}^{\mathcal{A}} \SE$ for all $t \in T$ where
\begin{itemize}
\item $v \in V^{\mathcal{A}(G)} $ if and only if $v \in V(G)$,
\item $uv \in E^{\mathcal{A}(G)} $ if and only if $uv \in E(G)$,
\item $t \in T^{\mathcal{A}(G)} $ if and only if $t \in T(G)$,
\item $(t_{1},t_{2}) \in R^{\mathcal{A}(G)}$ if and only if $t_{1},t_{2} \in T(G)$ and $t_{1} <_{T} t_{2}$,
\item $v \in L_{t}^{\mathcal{A}(G)}$ if and only if $t \in \zeta_{v}(v)$, and
\item $uv \in L_{t}^{\mathcal{A}(G)}$ if and only if $t \in \zeta_{e}(uv)$.
\end{itemize}

When the structure is understood from context, we will drop the $\mathcal{A}$ superscript.

\begin{theorem}
Given a dynamic graph $G$, if $G_{t}$ has local treewidth (effectively) bounded by $f(r)$ at each time $t \in T(G)$, then the Gaifman graph $\mathcal{G}$ obtained from $\mathcal{A}(G)$ has local treewidth (effectively) bounded by $\max\{f(r),\Card{T(G)}-1\}$.
\end{theorem}

\begin{proof}
Let $G_{t}$ be the graph at time $t\in T(G)$ and $\mathcal{G}_{t}$ the corresponding component of $\mathcal{G}$. $\mathcal{G}_{t}$ is essentially $G_{t}$.

As $G_{t}$ has (effectively) bounded local treewidth, there is a function $f$ such that for each $v\in V$ and $r\in \mathbb{N}$ we have $\tw([N_{r}(v)]) \leq f(r)$ in $G_{t}$ (note that this may not be true if we consider $G$ in total, ignoring the timing of the edges). In particular for each $v$ and $r$ in each $G_{t}$, we have a tree decomposition where each bag has at most $f(r)+1$ elements. Moreover we have such a decomposition at each time $t$, with no interaction between any two snapshots. Therefore the width of the decomposition is at most $f(r)$.

The only component of $\mathcal{G}$ that does not have an appropriate decomposition is the clique defined by the time elements representing $T(G)$. The only possibly decomposition is to place all elements in one bag, with width $\Card{T(G)}-1$.

As each $\mathcal{G}_{t}$ and the clique are disjoint, there are no additional edges to consider in the possible decompositions.
\end{proof}

We note also that if the graph has bounded degree at every $t \in T(G)$, this property is also preserved in the Gaifman graph. With the construction of Section~\ref{sec:ordered_time_struct}, planarity can also be preserved.

\subsection{Treewidth Preserving Structure}

If we are only concerned with the treewidth of the structure, we can use somewhat more natural structure.

Let $G=(V,E,T,\zeta_{v},\zeta_{e})$ be a dynamic graph.

In this case the universe $A$ consists of
\begin{itemize}
\item an element $v^{t}$ for each $v \in (V)$ if $t \in \zeta_{v}(v)$.
\item an element for each $t \in T(G)$,
\end{itemize}
and $\tau = \SB V^{\mathcal{A}}, L_{v}^{\mathcal{A}}, \varXi^{\mathcal{A}}, T^{\mathcal{A}}, R^{\mathcal{A}} \SE$ where
\begin{itemize}
\item $v^{t} \in V^{\mathcal{A}(G)}$ if and only if $v \in V(G)$ and $t \in \zeta_{v}(v)$,
\item $v^{t} \in L_{v}^{\mathcal{A}(G)}$ if and only if $v^{t}$ is generated from $v$,
\item $(u,v,t) \in \varXi^{\mathcal{A}(G)}$ if and only if $uv \in E(G)$ and $t \in \zeta_{e}(uv)$,
\item $t \in T^{\mathcal{A}(G)}$ if and only if $t \in T(G)$, and
\item $(t_{1},t_{2}) \in R^{\mathcal{A}(G)}$ if and only if $t_{1},t_{2} \in T(G)$ and $t_{1} <_{T} t_{2}$. 
\end{itemize}

\begin{theorem}
If $G$ has $\tw(G) \leq k$ at every time $t \in T(G)$ then the Gaifman graph $\mathcal{G}$ obtained from $\mathcal{A}(G)$ has $\tw(\mathcal{G}) \leq \max\{k+1,\Card{T(G)}-1\}$.
\end{theorem}

\begin{proof}
Let $G_{t}$ be the snapshot of $G$ at time $t\in T(G)$ and let $\tw(G_{t}) \leq k$ for every $t \in T(G)$. The Gaifman graph $\mathcal{G}$ consists of $\Card{T(G)}+1$ components; a component $\mathcal{G}_{t}$ for each $G_{t}$ which is the normal Gaifman graph translation of a graph, i.e. the graph itself, and a $(\Card{T(G)})$-clique  corresponding to the time elements of $A$. Each element of the time clique is connected to every vertex in exactly one $\mathcal{G}_{t}$ (precisely the element corresponding to time $t \in T(G)$) and each vertex in each $\mathcal{G}_{t}$ is connected to exactly one vertex of the time clique.

By assumption for each $\mathcal{G}_{t}$ we have $\tw(\mathcal{G}_{t}) \leq k$. Thus each $\mathcal{G}_{t}$ has a tree decomposition where each bag has at most $k+1$ elements. To each of these bags we add the element corresponding to time $t \in T(G)$. The time clique forms a bag of size $\Card{T(G)}$. We add an edge to the tree decomposition between the time clique bag and an arbitrarily chosen bag from each decomposition of each $\mathcal{G}_{t}$. Thus we have a tree decomposition of $\mathcal{G}$ with bag size at most $\max\{k+2,\Card{T(G)}\}$ and therefore $\tw(\mathcal{G}) \leq \max\{k+1,\Card{T(G)}-1\}$.
\end{proof}

\subsection{Structures for Totally Ordered Time}\label{sec:ordered_time_struct}

If we can assume that time is linear for our dynamic graph $G$, i.e. there exists a total order on the elements of $T(G)$, we may construct a structure where the Gaifman graph avoids the clique created by the elements corresponding to the times. The trade-off is that the contruction of logical sentences becomes more involved in the sense that the sentences become longer.

To obtain this modified structure we restrict the relation $R$ in each $\tau$ such that $(t_{1},t_{2}) \in R$ if and only if $t_{1}$ immediately precedes $t_{2}$. Furthermore we add a constant $s$ such that $(s,t_{i}) \in R$ if and only if $t_{i}$ is the earliest time element. Then the component of the Gaifman graph constructed from $\mathcal{A(G)}$ corresponding to $R$ is a path, where each $t_{i}$ subtends a vertex, and there is an edge between two such vertices if and only if one immediately precedes the other temporally. Moreover we can define the distance between each time element and $s$ recursively:
\begin{eqnarray*}
d_{1}(t_{i}) &:=& Rst_{i}\\
d_{n}(t_{i}) &:=& \exists t_{j} (Tt_{j} \wedge Rt_{i}t_{j} \wedge d_{n-1}(t_{j})) 
\end{eqnarray*}
Note of course that as first order logic is compact, these formul\ae{} must be defined separately for each $\Card{T(G)}$, giving a finite set in each case. As $\Card{T(G)}$ will always be taken as a parameter we effectively produce an infinite family of such formul\ae{}, and select the appropriate subset for each instance of whichever problem we deal with. In the case of second order logic this problem does not exist, but it is convenient to use the first order construction.
Then within the logic we can define the order relation over the time elements as:
\begin{eqnarray*}
t_{i} \leq t_{j} & := & (t_{i} = t_{j}) \vee \bigvee_{l \in [1,\Card{T(G)}]} (d_{l}(t_{i}) \wedge \neg d_{l}(t_{j}))
\end{eqnarray*}

The operator $\leq$ allows comparison over all times, but at the cost of formulae which may depend on the size of $T(G)$. However there is a notable change in treewidth and local treewidth of the structures using this approach.

\begin{theorem}
Given a dynamic graph $G$, let $\mathcal{A}(G)$ be the structure constructed from $G$ as a Local Treewidth Preserving Structure, except with a linear time relation $R$.

If $G$ has local treewidth (effectively) bounded by $f(r)$ at every $t\in T(G)$, then the Gaifman graph $\mathcal{G}$ derived from $\mathcal{A}(G)$ has local treewidth (effectively) bounded by $f(r)$.
\end{theorem}

\begin{proof}
The only change to the structure in this case is the component concerned with the time elements. In the first construction this was a clique of size $\Card{T(G)}$, but now is a path of length $\Card{T(G)}$. As trees have treewidth $1$, and local treewidth is bounded by treewidth, the local treewidth of this component is at most $1$.
\end{proof}

\begin{theorem}
Given a dynamic graph $G$, let $\mathcal{A}(G)$ be the structure constructed from $G$ as Treewidth Preserving Structure, except with a linear time relation $R$.

If $G$ has $\tw(G)\leq k$ at every $t\in T(G)$, then the Gaifman graph $\mathcal{G}$ derived from $\mathcal{A}(G)$ has $\tw(\mathcal{G}) \leq k+1$.
\end{theorem}

\begin{proof}
As before each component corresponding to the graph $G_{t}$ at time $t$ has treewidth at most $k$, and as before we add $t$ to each bag of this decomposition. For the time component we construct the following tree decomposition:

\begin{enumerate}
\item Each vertex $t$ is given a bag labelled $\{t\}$,
\item Each edge $st$ is given a bag labelled $\{s,t\}$ and
\item each vertex bag is connected in the decomposition to the two edge bags that contain the same label.
\end{enumerate}

This component clearly has treewidth $1$. Then we complete the decomposition for $\mathcal{G}$ by adding a decomposition edge from the vertex bag for time element $t$ to an arbitrary bag in the decomposition of $G_{t}$. Clearly this is still a tree, and the whole decomposition has width at most $k+1$.
\end{proof}

Note that for both representations of the temporal element of the graph, the complexity of the problem is at least partially dependent on $\Card{T(G)}$, in that either the (local) treewidth depends on $\Card{T(G)}$, or the length of any formula where we have to compare times does.

\section{Applications to Dynamic Graph Problems}

\subsection{Adapting the Metatheorems to the Dynamic Context}

\begin{theorem}\label{thm:tw_MC}
Let $G$ be a dynamic graph with $\tw(G) \leq k$ at every time $t$ and $\phi$ a sentence of monadic second order logic. The problem \textsc{MC($G$,$\phi$)} is fixed-parameter tractable with parameter $k+\Card{\phi}+\Card{T(G)}$.
\end{theorem}

\begin{proof}
Courcelle's Theorem gives us that the model checking problem for graphs of bounded treewidth is fixed-parameter tractable with the treewidth and the length of the formula as a combined parameter.

In our case we may apply Courcelle's theorem to the Gaifman graph derived from the dynamic graph. As the treewidth of the Gaifman graph is bounded by the treewidth of the original graph and $\Card{T(G)}$, we obtain our result. Note that in the case where we use the linear time structure variant, $\Card{T(G)}$ is implicitly included in $\Card{\phi}$, if necessary.
\end{proof}

\begin{theorem}\label{thm:ltw_MC}
Let $G$ be a dynamic graph with effectively bounded local treewidth at every time $t$ and $\phi$ a sentence of first order logic. The problem \textsc{MC($G$,$\phi$)} is fixed-parameter tractable with parameter $\Card{\phi}+\Card{T(G)}$.
\end{theorem}

\begin{proof}
By Frick and Grohe's Theorem, the model checking problem for graphs of effectively bounded local treewidth is fixed parameter tractable with the length of the formula as the parameter. Again the length of the formula may implicitly depend upon other parameters as part of the problem.

Stewart notes that Frick and Grohe's theorem still holds if the bound on the local treewidth depends upon a parameter, rather than simply being constant. In the case where we do not use the linear time construction, the local treewidth is bounded by $\max\{f(r),\Card{T(G)}-1\}$.
\end{proof}

We demonstrate the use of these theorems by application to some dynamic graph problems of general interest. Bhadra and Ferreira~\cite{BhadraFerreira03} prove that the problem of finding a connected component of size at least $k$ in an evolving digraph is $\NP{}$-complete. 

\pcproblem{Strongly Connected Dynamic Component (SCDC)}{A dynamic digraph $D=(V,A,T,\zeta)$, an integer $k$.}{$k+\Card{T}$.}{Is there a set $V'\subseteq V$ with $\Card{V'}\geq k$ such that for every pair $u,v \in V'$ we have $\mathcal{J}(u,v)$?}

Adapting the structures for directed graphs is simply a matter of semantics for the $E$ relation in $\tau$. Using this we apply the metatheorem to demonstrate that the problem is fixed-parameter tractable.

\begin{lemma}\label{lem:SCDC_exp}
\textsc{SCDC} is expressible in first order logic.
\end{lemma}

\begin{proof}
We first define a sentence that expresses the idea of a journey.

\begin{eqnarray*}
\mathcal{J}_{n}(u,v) & := & \mathcal{J}_{n-1}(u,v) \vee\\
&&\exists x_{1},\ldots,x_{n+1}, t_{1}, \ldots, t_{n}(\bigwedge_{i \in [1,n+1]}Vx_{i} \wedge \bigwedge_{i \in [1,n]}Tt_{i})\wedge\\
&&(\bigwedge_{i\in [1,n]} Ex_{i}x_{i+1} \wedge L_{t_{i}}x_{i}x_{i+1})\wedge\\
&&(\bigwedge_{i\in[1,n]}t_{i}\leq t_{i+1}) \wedge (f_{V}(u) = f_{V}(x_{1}) \wedge f_{V}(v) = f_{V}(x_{n+1})).
\end{eqnarray*}

Using this the problem can be succintly defined in first order logic as:

\[
\exists v_{1},\ldots,v_{k}\forall x,y (f_{V}(x) = f_{V}(v_{i}) \wedge f_{V}(y) = f_{V}(v_{j}) \Rightarrow \mathcal{J}_{k}(x,y))
\]

Note that the length of the sentence depends upon only $k$ and $\Card{T(G)}$ and that we again avoid running afoul of the compactness of first order logic by restricting the possible journey lengths to $k$ for each instance.
\end{proof}

Combining Theorem~\ref{thm:ltw_MC} and Lemma~\ref{lem:SCDC_exp} we obtain the tractability result.

\begin{corollary}
\textsc{SCDC} is fixed-parameter tractable for dynamic graphs of bounded local treewidth.
\end{corollary}

As first order logic is a subset of monadic second order logic we may also use Theorem~\ref{thm:tw_MC}.

\begin{corollary}
\textsc{SCDC} is fixed-parameter tractable for dynamic graphs of bounded treewidth.
\end{corollary}

Another interesting dynamic graph problem is whether a sending vertex can receive a reply without the reply having to travel too far.

\pcproblem{Short Message Return Path (SMRP)}{A dynamic digraph $D=(V,A,T,\zeta)$ with an identified vertex $v\in V$ and an integer $k$.}{$k+\Card{T}$.}{Is there a journey from each $u \in N^{out}(v)$ to $v$ of length at most k?}

\begin{lemma}\label{lem:SMRP_exp}
\textsc{SMRP} is expressible in first order logic.
\end{lemma}

\begin{proof}
Using the ability to express the idea of a journey in first order logic,  \textsc{SMRP} is simple to express.
\[
\forall u (Vu \wedge Evu \Rightarrow \exists u' (f_{V}(u)=f_{V}(u) \wedge \mathcal{J}_{k}(u,v)))
\]
\end{proof}

Using Theorems~\ref{thm:ltw_MC} and~\ref{thm:tw_MC} with Lemma~\ref{lem:SMRP_exp} we obtain the following tractability result.

\begin{corollary}
\textsc{Short Message Return Path} is fixed-parameter tractable for dynamic graphs of bounded local treewidth and bounded treewidth.
\end{corollary}

\subsection{Transferring Static Results}

The transference of Courcelle's and Frick \& Grohe's metatheorems also imply the transference of previous results on static graph, with only minor changes in problem formulation. Most immediately, these results hold where we relax requirements such that the edges required for the solution exist at some point.

We demonstrate by ``temporalizing'' \textsc{$k$-Coloring} and the associated MSO formula that demonstrates it fixed-parameter tractability for graphs of bounded treewidth.

\pcproblem{Permanent Coloring}{A dynamic graph $G=(V,E,T,\zeta)$, an integer $k$.}{$k+\Card{T}$.}{Is there a static assignment of $k$ colors to $V$ that is a proper coloring of $G$ at each time $t\in T$?}

\begin{lemma}
\textsc{Permanent Coloring} is expressible in monadic second order logic.
\end{lemma}

\begin{proof}
\[
\begin{array}{l}
\exists V_{1},\ldots,V_{k}\forall x,y,t(\bigvee_{i\in[1,k]}\bigwedge_{j\neq i}(V_{i}x\wedge\neg V_{j}x)\wedge\bigvee_{i\in[1,k]}\bigwedge_{j\neq i}(V_{i}y\wedge\neg V_{j}y)\wedge\\
Vx\wedge Vy\wedge Tt \wedge Exyt \Rightarrow \bigwedge_{i\in[1,k]} \neg (V_{i}x\wedge V_{i}y))
\end{array}
\]
\end{proof}

Alternatively we can reframe the problem as requiring a solution that is true at every point in time.

\pcproblem{Evolving Coloring}{A dynamic graph $G=(V,E,T,\zeta)$, an integer $k$.}{$k+\Card{T}$.}{Is there a (possibly different) proper $k$-coloring of $G$ at each time $t\in T$?}

\begin{lemma}
\textsc{Evolving Coloring} is expressible in monadic second order logic.
\end{lemma}

\begin{proof}
\[
\begin{array}{l}
\forall t \exists V_{1},\ldots,V_{k}\forall x,y(\bigvee_{i\in[1,k]}\bigwedge_{j\neq i}(V_{i}x\wedge\neg V_{j}x)\wedge\bigvee_{i\in[1,k]}\bigwedge_{j\neq i}(V_{i}y\wedge\neg V_{j}y)\wedge\\
Vx\wedge Vy\wedge Tt \wedge Exyt \Rightarrow \bigwedge_{i\in[1,k]} \neg (V_{i}x\wedge V_{i}y))
\end{array}
\]
\end{proof}

Then by Theorem~\ref{thm:tw_MC}:

\begin{corollary}
\textsc{Permanent Coloring} and \textsc{Evolving Coloring} are fixed parameter tractable on graphs of bounded treewidth.
\end{corollary}

It is easy to see that, although the formul\ae{} become somewhat more complex, it is a simple matter to translate existing expressibility results (often already in the context of parameterized tractability) into results for dynamic graphs. Thus the permanent and evolving versions of interesting problems such as \textsc{Dominating Set}, \textsc{Independent Set}, \textsc{Vertex Cover}, \textsc{Clique} and \textsc{Subgraph Isomorphism} are fixed-parameter tractable on dynamic graphs of bounded (local) treewidth (Flum \& Grohe~\cite{FlumGrohe06} give appropriate static results).

\section{Dynamic Graph Classes with Bounded (Local) Treewidth}

One broad class of dynamic graphs that has proven interesting due to the connectivity problems induced are sparse dynamic graphs. There are several definitions of what constitutes a sparse graph in the dynamic context, with two common definitions involving Markovian processes and bounded degree graphs. Both of these classes in fact have bounded local treewidth.

Flum \& Grohe~\cite{FlumGrohe06} note that given a graph $G$ of maximum degree $d$ we have $\ltw(G,r) \leq d^{r}$. Thus the class of bounded degree dynamic graphs has bounded local treewidth.

As mentioned, sparse graphs may also be defined by Markovian processes on the edges of the graph where the probability of an edge existing, and continuing to exist are bounded (e.g.,~\cite{ClementiMS11}). This also leads to bounded local treewidth for such graphs.

\begin{theorem}
Let $G=(G_{1},\ldots,G_{t})$ be a edge-Markovian dynamic graph where the probability of a non-edge becoming an edge is $p \leq c_{1}/n$ and the probability of an edge becoming a non-edge is $q \leq 1-c_{2}/n$, then the expected maximum degree for each $G_{i\geq 2}$ is at most $c_{1}+c_{2}$.
\end{theorem}

\begin{proof}
At any time $t$, for each vertex $v \in V(G)$ we have
\begin{eqnarray*}
\Card{N_{t+1}(v)} &=& (1-q)\cdot\Card{N_{t}(v)} + p(n-1-\Card{N_{t}(v)})\\
&\leq& \frac{c_{2}}{n}\cdot\Card{N_{t}(v)}+\frac{c_{1}}{n}\cdot n -\frac{c_{1}}{n} \cdot \Card{N_{t}(v)} - \frac{c_{1}}{n}\\
&=& c_{2}\cdot\frac{\Card{N_{t}(v)}}{n} + c_{1} -c_{1}\frac{\Card{N_{t}(v)}+1}{n}\\
&\leq& c_{1}+c_{2}
\end{eqnarray*}
\end{proof}

\begin{corollary}
Let $G=(G_{1},\ldots,G_{t})$ be a edge-Markovian dynamic where the probability of a non-edge becoming an edge is $p \leq c_{1}/n$ and the probability of an edge becoming a non-edge is $q \leq 1-c_{2}/n$, then $G$ is expected to have bounded local treewidth.
\end{corollary}

We note that the third common definition for sparse dynamic graphs, using some measure of density of the graph, does not lead to bounded local treewidth, and in fact these graphs may have (local) treewidth on the order of at least $\sqrt{\Card{V(G)}}$. Some measures of density lead to bounded local treewidth, but essentially by implying bounded degree.

\section{Conclusion}

The structural parameters treewidth and local treewidth have proven to be very useful in algorithmic design when dealing with hard problems. Aiding in this task are the metatheorems of Courcelle and Frick \& Grohe that provide complexity classifications for large classes of graphs using logical expressibility. As dynamic graph theory and dynamic graph problems become more prominent thanks to the advance of technology, questions of treewidth arise. We have shown that given a finite, bounded temporal context the notions of logical expressibility and the corresponding metatheorems can be extended to provide tools for classification of problems on dynamic graphs. In particular if a dynamic graph has bounded (local) treewidth at every point in time, then we can construct a logical structure, and hence Gaifman graph where the (local) treewidth remains bounded. Furthermore in the finite, bounded temporal context first order and thus second order logic can be augmented (finitely) with relations that can express the temporal notion of before and after. With these tools we can easily classify many problems, including temporal extension of static problems.

A natural focus for extending this work would be to increase the graph classes that can be classified, ideally up to graphs of bounded expansion \emph{\`{a} la} Dvo\v{r}\'{a}k \emph{et al.}~\cite{DvorakKT10}, though of course we cannot do better unless it is also possible for static graphs.

This approach however is limited by the fact that we require the length of the time interval to be finite and a parameter. If the length of time is unbounded (even finite but unbounded), then we do not have the logical expressibility to deal with these problems. Thus certain classes of dynamic graphs are excluded by this condition. However we can still easily express properties over a finite amount of time, or in periodic graphs, which covers a significant amount of real applications, where finite time is a constraint of the setting, rather than an artificial stipulation of the theory.

\bibliographystyle{alpha}
\bibliography{logic}

\end{document}